\documentclass[12pt, tikz]{article}
\usepackage[margin=0.9in]{geometry} 

\usepackage{graphicx, xcolor}
\usepackage{multirow} 
\usepackage{float}
\usepackage{tikz}
\usepackage{ulem}
\usetikzlibrary{tikzmark}
\usepackage{natbib}

\usepackage{amsfonts}
\usepackage{transparent}
\usepackage{amsmath, scalerel}
\usepackage{bbm}
\usepackage{bm}
\usepackage{adjustbox}
\usepackage{ifthen}
\usepackage[ruled,linesnumbered]{algorithm2e}
\usepackage{amsthm, amssymb} 
\usepackage{thm-restate}

\definecolor{toc}{RGB}{13,55,174}	
\usepackage{hyperref}				
\hypersetup{
    colorlinks=true,
    citecolor=toc,
    filecolor=black,
    linkcolor=toc,
    urlcolor=toc
}

\allowdisplaybreaks 

\newtheorem{theorem}{Theorem}[section]

\newtheorem{proposition}[theorem]{Proposition}

\newtheorem{example}[theorem]{Example}

\newcommand{\dist}{\mathcal{D}}
\newcommand{\boxes}{\mathcal{B}}
\newcommand{\reals}{\mathbb{R}}
\newcommand{\lp}{\left}
\newcommand{\rp}{\right}
\newcommand{\E}[2]{\mathbb{E}_{#1}\lp[#2 \rp]}
\renewcommand{\Pr}[2]{\textbf{Pr}_{#1}\lp[#2 \rp]}

\newcommand{\s}{\sigma}
\newcommand{\alg}{\text{ALG}}

\newcommand{\weitz}{\text{WEITZ}}

\date{}

\title{The Role of Commitment in Optimal Stopping}

\author{Jos\'e Correa \\ Universidad de Chile \and
Evangelia Gergatsouli\footnote{Part of the work was done when this author was a Visiting PhD in Center for Mathematical Modeling (CMM) at University of Chile, in March-May 2023.}\\ Georgia Institute of Technology \and  Bruno Ziliotto\footnote{This work was supported by the French Agence Nationale de la Recherche (ANR) under reference ANR-21-CE40-0020 (CONVERGENCE project) and under grant ANR-17-EURE-0010 (Investissements d'Avenir program). Part of this work was done during a 1-year visit of this author to the Center for Mathematical Modeling (CMM) at University of Chile in 2023, under the IRL program of CNRS.} \\ Toulouse School of Economics}

\begin{document}

\maketitle

 We investigate the role of commitment in optimal stopping by studying all the variants between  Prophet Inequality (PI) and Pandora's Box (PB). Both problems deal with a set of variables drawn from known distributions.
 In PI the gambler observes an adversarial order of these variables with the goal of selecting one that maximizes the expected value against a prophet who knows the exact values realized. The gambler has to irrevocably decide at each step whether to select the value or discard it (commitment). 
 On the other hand, in PB the gambler selects the order of inspecting the variables and for each pays an observation cost to see the actual value realized, aiming to choose one to maximize the net cost of the value chosen minus the observation cost paid. The gambler in PB can return and select any variable already seen (no commitment). 
 
 For all the variants between these problems that arise by changing parameters such as (1) commitment (2) observation cost (3)  order selection, we concisely summarize the known results and fill the gaps of variants not yet studied. We also uncover connections to Ski-Rental, a classic online algorithm problem.

\textbf{Keywords}:
Prophet Inequalities, Pandora's Box, optimal stopping, Ski-Rental, online algorithms
\setcounter{page}{0}
\thispagestyle{empty}
\newpage




\section{Introduction}\label{sec:intro}
In the Prophet Inequality problem a gambler is sequentially presented with $n$ independent samples $X_1, X_2, \ldots, X_n$ drawn from known distributions $\dist_1, \dist_2, \ldots, \dist_n$ and she has to irrevocably decide at each step to either select the sample presented to her and stop, or continue and lose the sample just seen. The gambler's reward is  $\E{}{X_\tau}$ where $\tau$ is the time she decides to stop. Her goal is to compete with an all-knowing prophet who can see the values of the samples in advance and choose their reward with this information; therefore, the prophet receives a reward $\E{}{\max_{i\in [n]} X_i}$. 
The classic result in prophet inequalities, shown in~\cite{KrenSuch1977,KrenSuch1978,Samu1984} states that it is possible for the gambler to receive $\E{}{X_\tau} \geq 1/2 \cdot \E{}{\max_i X_i}$.

A similar well-studied problem in economics and optimal stopping theory is Pandora's Box. In this problem, the gambler again sees $n$ independent samples from known distributions $\dist_1, \dist_2, \allowbreak \ldots, \dist_n$, but, in contrast to Prophet Inequality, she can now decide the order of exploring them. In order to see the value of a sample (or ``box"), the gambler has to pay an observation cost (opens the box), and after observing any set of the samples she chooses, she can decide to stop and take any of the samples observed (not necessarily the last seen). The goal is to maximize the best sample value seen minus the total observation cost paid i.e. $\E{}{\max_{i\in \mathcal{O}} X_i} - \sum_{i\in \mathcal{O}} c_i$, where $\mathcal{O}$ is the set of boxes opened. In this case, the gambler compares her expected gain with the online optimal who, unlike the Prophet, cannot see the values instantiated, while the optimal policy is described by Weitzman in~\cite{Weit1979}. Two natural questions arise at this point.
\begin{enumerate}
    \item In Prophet Inequality: what happens if the gambler is not committed to pick the last sample seen, but can turn back and pick an earlier one? \textcolor{black}{In this setting, the gambler can do as well as the Prophet, since he can wait} until the end of the sequence and pick the best one seen. A more natural setting would be to have this problem with an observation cost for each value.
    \item In Pandora's Box: how well does the gambler do against a prophet that knows all the values in advance? This setting is also trivial, since the prophet will only open one box, while the gambler will have to open boxes in the optimal policy described by Weitzman~\cite{Weit1979}. \textcolor{black}{Consequently, the gambler can not compete with the Prophet}. Therefore a more natural setting would be to fix the order at the beginning for both the gambler and the prophet. 
\end{enumerate}

More generally, between Pandora's Box and Prophet Inequalities, there are three parameters that characterize the problems: 
\begin{enumerate}
    \item \emph{Commitment}: whether the gambler can only select the last sample value seen or can pick an option seen earlier in the game.
    \item \emph{Observation cost}: whether the gambler has to pay an extra cost to observe the sample from each distribution.
    \item \emph{Order selection}: whether the gambler has control over the order of observation or it is adversarially presented to her.
\end{enumerate}

Altering each of these parameters leads to a different problem, both from a maximization and from a minimization perspective. We outline below in Table~\ref{table:maximization_summary} and Table~\ref{table:minimization_summary} the versions of the problems that arise for all combinations of these parameters, and which is the result known in each case. The results in \textbf{bold} are presented in this paper.

\begin{table}[H]
\centering
\begin{tabular}{c}
    \textbf{\large MAXIMIZATION} \\ 
\end{tabular}
\vspace{10pt}

\fbox{\textbf{Commitment: YES}}
\vspace{8pt} 
\begin{tabular}{ccccc}
\multicolumn{1}{c|}{\textbf{Best Ratio}}       & $\bm{-\infty}$   & $[0.726, 0.745]$  & $\bm{-\infty}$   & 1/2    \\ 
\multicolumn{1}{c|}{\textbf{Reference/Result}}        & {\ref{example:max_knows_values} }   & {  \cite{BubnaChipl2022, CorrFoncHoekOostVred2017} }   &  {~\ref{example:max_knows_values},\cite{Samu1992} }  & { \cite{KrenSuch1978} }   \\ \cline{1-5}
\multicolumn{1}{c|}{\textbf{Order Selection}}  & Yes & Yes & No  & No    \\ 
\multicolumn{1}{c|}{\textbf{Observation Cost}} & Yes & No  & Yes & No  \\
\end{tabular}
\bigskip

\fbox{\textbf{Commitment: NO}}\\
\vspace{8pt} 
\begin{tabular}{ccccc}
\multicolumn{1}{c|}{\textbf{Best Ratio}}       & $\bm{-\infty}$   & \textbf{1}   & $\bm{-\infty}$   & \textbf{1}    \\ 
\multicolumn{1}{c|}{\textbf{Reference/Result}}        &  {\small \ref{example:max_knows_values} }   &   {\small Trivial~\ref{ex:trivial_max} }   &  {\small \ref{example:max_knows_values}}    & {\small Trivial~\ref{ex:trivial_max}}   \\ \cline{1-5}
\multicolumn{1}{c|}{\textbf{Order Selection}}  & Yes & Yes & No  & No        \\ 
\multicolumn{1}{c|}{\textbf{Observation Cost}} & Yes & No  & Yes & No                          
\end{tabular}

\caption{Best known ratios for all eight variants of the maximization version of the problem. In the case of commitment, no observation cost and order selection (\emph{free order prophet}) the lower bound $0.7258$ is from~\cite{BubnaChipl2022} and the upper bound $1/1.342\approx 0.745$ is from~\cite{CorrFoncHoekOostVred2017}.}

\label{table:maximization_summary}
\end{table}
\normalsize


\begin{table}[H]
\centering
\begin{tabular}{c}
    \textbf{\large MINIMIZATION} \\ 
\end{tabular}
\vspace{10pt}

\fbox{\textbf{Commitment: YES}}
\vspace{8pt} 
\begin{tabular}{ccccc}
\multicolumn{1}{c|}{\textbf{Best Ratio}}       & $\infty$   & $\infty$ & $\infty$   & $\infty$    \\ 
\multicolumn{1}{c|}{\textbf{Reference/Result}}        & B.1 in \cite{LivaMeht2022}  & B.1 in \cite{LivaMeht2022}   & B.1 in \cite{LivaMeht2022}  & B.1 in \cite{LivaMeht2022}   \\ \cline{1-5}
\multicolumn{1}{c|}{\textbf{Order Selection}}  & Yes & Yes & No  & No    \\ 
\multicolumn{1}{c|}{\textbf{Observation Cost}} & Yes & No  & Yes & No  \\
\end{tabular}
\bigskip

\fbox{\textbf{Commitment: NO}}\\
\vspace{8pt} 
\begin{tabular}{ccccc}
\multicolumn{1}{c|}{\textbf{Best Ratio}}       & $\infty$   & \textbf{1 }  & \textbf{1.58*}& \textbf{1}    \\ 
\multicolumn{1}{c|}{\textbf{Reference/Result}}        &  \ref{ex:MIN},\cite{Weit1979}   &  Trivial~\ref{ex:trivial_max}  &  \ref{ex:ski_rental}   & Trivial~\ref{ex:trivial_max}   \\ \cline{1-5}
\multicolumn{1}{c|}{\textbf{Order Selection}}  & Yes & Yes & No  & No        \\ 
\multicolumn{1}{c|}{\textbf{Observation Cost}} & Yes & No  & Yes & No                          
\end{tabular}

\caption{Best known ratios for all eight variants of the minimization version of the problem. The ratio of $1.58$ noted with a \textbf{*} also holds for a weaker prophet that knows the values but not the order of arrival.
}
\label{table:minimization_summary}
\end{table}

In Section~\ref{sec:prelims}, we provide a brief overview of related problems. Then, in Sections~\ref{sec:max_examples} and \ref{sec:min_examples}, we present the results from Table~\ref{table:maximization_summary} for the maximization version and Table~\ref{table:minimization_summary} for the minimization version, respectively.

\section{Preliminaries}\label{sec:prelims}

We briefly describe the two main problems, Prophet Inequality and Pandora's Box, which motivated us to study their variations presented in Tables~\ref{table:maximization_summary} and ~\ref{table:minimization_summary}. We also give a brief overview of Ski-Rental, a problem that comes up in the minimization variant for the case with observation cost but without order selection.

\paragraph{\textbf{Prophet Inequality}~\cite{LivaMeht2022}} 

In the cost minimization Prophet Inequality we are given $n$ distributions $\mathcal{D}_1, \ldots, \mathcal{D}_n$ and at every step $t$ we see a sample drawn from $X_t\sim \mathcal{D}_t$. We have to decide to either take it or continue and lose it forever. Our goal is to pick the minimum sample compared to a Prophet who can see all samples in advance, i.e. 
\[
\text{Prophet} = \E{\dist}{\min_{i\in[n]}X_i}.
\]

\paragraph{\textbf{Pandora's Box}~\cite{Weit1979,BeyhCai2023}} 
In Pandora's Box we are given a set  $\boxes$ of $n$ boxes, each with a known opening cost $c_b\in\reals_+$, and containing an unknown value $v_b$ drawn from a known distribution $\dist_b$. Each box $b\in \boxes$, once it is opened, reveals the value $v_{b}$.

The algorithm can open boxes sequentially, by paying the opening cost each time, and observe the value instantiated inside the box. The goal of the algorithm is to choose a box of small value, while spending as little cost as possible ``opening" boxes. Formally, denoting by $\mathcal{O} \subseteq \boxes$ the set of opened boxes, we want to minimize
\[\E{v\sim \dist}{ \min_{b\in \mathcal{O}} v_b + \sum_{b\in \mathcal{O}}c_b}.\]

\paragraph{Weitzman's Algorithm for Pandora's Box}
In the minimization Pandora's Box, \cite{Weit1979} described a greedy algorithm that achieves the optimal (on expectation) cost. The core of this algorithm is the calculation of an index $\s_b$ for every box $b$, called \emph{reservation value}, formally defined as the value that satisfies the following equation
\begin{equation}\label{eq:reservation_value}
    \E{v \sim \mathcal{D}_b}{(\s_b-v)^+} = c_b,
\end{equation}

where $(a-b)^+ = \max(0, a-b)$. Then, the boxes are ordered by increasing $\sigma_b$ and opened until the minimum value revealed is less than the $\s_b$ of next box in the order. The algorithm is formally described in Algorithm~\ref{alg:weitzman_main}.\\

\begin{algorithm}[H]
		\KwIn{Boxes with costs $c_b\in \mathbb{R}^+$, distribution  $\mathcal{D} = (\mathcal{D}_1, \mathcal{D}_2, \ldots, \mathcal{D}_n)$.}

       An unknown vector of values $\bm{v} \sim \mathcal{D}$ is drawn\\
        Calculate $\sigma_b$ for each box $b$ using equation~\eqref{eq:reservation_value} and rename boxes such that $\s_1\leq \s_2\leq  \ldots \leq \s_n$\\
$V_{\textbf{min}} \leftarrow \infty$\\
	\For{every box $i=\{1, \ldots, n\}$} {
		 Pay $c_i$, open box $i$ and observe $V_i$\\
			$V_{\textbf{min}} \leftarrow \textbf{min}(V_{\textbf{min}}, V_i)$\\
			\If{$V_{\textbf{min}} < {\s}_{i+1}$}{
				Stop and collect $V_{\textbf{min}}$\\
   }
}
		\caption{Weitzman's algorithm.}\label{alg:weitzman_main}
	\end{algorithm}

\paragraph{\textbf{Ski-rental}~\cite{KarlManaMcgeOwic1990}} In the ski rental problem, a gambler is going skiing for an unknown number of days (e.g depending on the weather). Since she does not own skis, every day she decides to either rent skis at a cost of 1, or buy them at a cost of $B$ and incur no further costs. The prophet knows exactly when the last ski day $T$ will be, and therefore can decide up front if it is cheaper to buy at the beginning or only rent; their cost is
\[\text{Ski-prophet} = \min_{i\in [T]} (i+X_i) = \min(T,B).\]

\section{Results for Maximization}\label{sec:max_examples}


\begin{proposition}\label{ex:trivial_max}
Whenever there is \textbf{no} commitment and \textbf{no} opening cost, the best ratio is 1. 
\end{proposition}
\begin{proof}
The decision-maker and the prophet have the same optimal algorithm: open everything, without cost, and choose the best sample. This is true for both minimization and maximization.
\end{proof}
\begin{example}[Ratio is $-\infty$]\label{example:max_knows_values}
The following example works for the cases of the maximization version where there is an observation cost. The variables are all i.i.d. for $k=1,\ldots, n$
\[X_k = \begin{cases}
  n  & \text{w.p. } 1/n^{1.5}\\
   0 & \text{w.p. }1-1/n^{1.5}.
\end{cases}
\]
 Denote by $p=1-n^{-1.5}$. When only 0 are realized the optimal will just stop at the first box (paying -1), otherwise the gain is $n$ and at each step $i$ the probability of stopping is $\Pr{}{\text{stop at step }i}=p^{i-1} (1-p)$. Therefore we calculate the gain of the prophet as 
\begin{align*}
\text{Prophet} & = p^n(-1)  +  (1-p^n)n - \sum_{i=1}^n i\cdot \Pr{}
{\text{stop at step }i} & \\
& = n - p^n - \frac{1-p^n}{1-p} & \text{From }~\ref{eq:sum_max_p}\\
& = n-n^{1.5}+ p^n\lp(n^{1.5}-1 \rp) & \text{Def of }p \text{ \& rearrange}\\
& \approx n - n^{1.5} + \lp(1-n^{-0.5} - \frac{1}{2n} \rp)(n^{1.5}-1) & \text{Exp limit and Taylor}\\
& = -2 + n^{-1.5} + \frac{1}{2 n} - \frac{n^{0.5}}{2} + n > 0
    \end{align*}

where in the $\approx$ above we used that $p^n = (1-n^{-1.5}) \approx e^{-n^{-0.5}} \approx 1-n^{-0.5}+1/2 n^{-1}$ from the exponential limit for $n$ large and Taylor approximation. While the algorithm will either try to find a $n$ or stop at the first box getting
\begin{align*}
\alg & = \max\lp\{-1+n(1-p), 
p^n(-n) +  (1-p^n) n - \sum_{i=1}^n \Pr{}{\text{Stop at step }i} i\rp\} \\
& = \max \lp\{-1 + n(1-p), n-np^n-\frac{1-p^n}{1-p}\rp\}\\
& = \max \lp\{-1+\frac{n}{n^{1.5}}, (1-p^n)(n-n^{1.5})\rp\} \text{Def of }p \text{ and rearrange}\\
& <0 \text{, as $n$ grows.}
\end{align*}

\noindent 
For both these calculations we used that 
\begin{align}\notag
\sum_{i=1}^n i\cdot \Pr{}{\text{stop at step }i} &= \sum_{i=1}^n i \cdot (1-p) p^{i-1}  \\ \notag
& = \frac{1-p^n-np^n+np^{n+1}}{1-p}\\ 
& = \frac{1-p^n}{1-p}-np^n.\label{eq:sum_max_p}
\end{align}





\textbf{Order selection.} In the versions where the order can be chosen, note that the cost of the gambler will remain the same since the variables are all i.i.d. If all variables are 0 the prophet still just opens the first and stops, otherwise she opens first one of the boxes that give $n$ and pay -1. Therefore 
\[\text{Prophet} = p^n(-1) + (1-p^n)(n-1) \approx -1+n(n^{-0.5}-1/2 n^{-1}) = n^{0.5} - 3/2 >0\]
which is the same as before, and the example still holds.

\textbf{Commitment vs no-commitment.} This example holds regardless of whether we have commitment or not. The only acceptable value is $n$; when the $n$ is found any reasonable algorithm stops immediately. If no $n$ is found, the algorithm just takes the last $0$ found, and there is no point returning and taking a different $0$.

\end{example}



\section{Results for Minimization}\label{sec:min_examples}
\begin{example}[Ratio is $\infty$]\label{ex:MIN}
Consider the following i.i.d. instance:
\[X_k = \begin{cases}
  0  & \text{w.p. } n^{-1/2}  \\
   n & \text{w.p. } 1-n^{-1/2},
\end{cases}
\]
for $k=1,\ldots n$, and cost 1. Whenever there is a 0, the Prophet will put it in first position, and pick it. Otherwise, he will pick the first variable. 
His payoff is then 
\[
\E{}{\text{Prophet}} = \lp(1-(1-n^{-1/2})^n\rp) \cdot 1 + (1-n^{-1/2})^n \cdot (n+1) = n \cdot \lp(1-n^{-1/2} \rp) + 1,
\]

\noindent
therefore $\E{}{\text{Prophet}} \rightarrow 1$ as $n\rightarrow \infty$. The gambler will wait for the first 0 to appear, and pick it incurring cost
\begin{align*}
    ALG & = \sum_{i=1}^n \Pr{}{\text{There is a 0 at }i} (i+0) + \Pr{}{\text{No 0 exists}}\cdot (n+n) \\
     & = \sum_{i=1}^n n^{-1/2}\cdot\lp(1 - n^{-1/2}\rp)^{i - 1}\cdot i + \lp(1 - n^{-1/2}\rp)^n \cdot 2n\\
     & = n^{1/2}  + \lp(1 - n^{-1/2}\rp)^n \lp(n-n^{1/2}\rp)\\
     & \approx n^{1/2}
\end{align*}
Therefore the ratio is $\alg/\text{Prophet} = n^{1/2} \rightarrow \infty$ as $n$ grows.
\end{example}

\begin{theorem}[Ratio is $\frac{e}{e-1}$]\label{ex:ski_rental} 
The Ski-Rental algorithm~\cite{KarlManaMcgeOwic1990} obtains $\frac{e}{e-1} \simeq 1.58$ and this ratio is tight.

\end{theorem}

\begin{proof}
    We first show the approximation ratio then the tightness.\\
    
\textbf{Approximation ratio.} Consider a variant of the Ski-rental problem called the \emph{Decreasing Buying cost Ski-rental (DB-Ski Rental)} defined as follows; we are going skiing for an unknown number of days and each day $t$ we can either rent the skis (for one day) paying $1$ or buy them for the rest of the ski season for $a_t$, where $a_t \geq a_{t+1}$ for all times $t\geq 1$. The prophet in this problem can see the whole sequence, and therefore pays $\min_{i\leq t} (i+a_i)$. The traditional Ski-rental of Section~\ref{sec:prelims} is a special case of this variant, where $a_t=B$ until the final (unknown) ski day, then $a_t=0$. 

Observe that DB-Ski-rental is exactly the same as the optimal selection problem when there is observation cost but no commitment (i.e. the version of this theorem). To see this define $a_t = \min_{i\leq t} X_i$ i.e. the buying cost at step $t$ is the minimum of all the variables seen so far, which is a decreasing sequence (as DB-Ski Rental requires). Any decision made by the DB-Ski-rental algorithm, directly translates to a decision in our problem; 
\begin{itemize}
    \item \textbf{If} DB-Ski-Rental rents (cost is 1): we inspect the next variable, paying 1.
    \item \textbf{Else} DB-Ski-Rental buys (cost is $a_t$): we stop and select $X_p$ where $p = \text{argmin}_{i\leq t} X_i$, and pay the exact same cost. At this step we crucially use that there is no commitment, and we are allowed to turn back and select a variable (potentially) seen earlier. 
\end{itemize}

Using Lemma 3.1 from~\cite{ChawGergTengTzamZhan2020}, which we copy below for convenience, we get that there exists an algorithm that obtains $\frac{e}{e-1}$ for our problem.

\textbf{Lemma 3.1 from~\cite{ChawGergTengTzamZhan2020}} \emph{Given any sequence $a_1\geq a_2 \geq \ldots$, there exists an online algorithm that chooses a stopping time $t$ so that}\footnote{You may notice that in this lemma the time is shifted by $1$ (i.e. we have $t-1$ instead of $t$). This is a technicality to allow the player to stop before renting, by having an immediate outside option of cost $a_1$. In our problem we set $a_1=\infty$, to avoid allowing the player to stop without choosing anything, so that now $X_1$ corresponds to $t=2$ and so on.}  
\begin{equation}\label{eq:ski_rental_FOCS}
    t-1+a_t \leq \frac{e}{e-1} \min_{j} (j-1+a_j)
\end{equation}

Furthermore, in Lemma 7.1 of \cite{ChawGergTengTzamZhan2020} the authors extend Lemma 3.1 to the case where the renting costs $p_i$ for each day $i$ are not all $1$ and equation \ref{eq:ski_rental_FOCS} becomes
\begin{equation}
        a_t + \sum_{i=1}^{t-1} p_i \leq \frac{e}{e-1} \min_{j} \lp(a_j + \sum_{i=1}^{j-1}p_i\rp).
\end{equation}

Therefore, the factor of $\frac{e}{e-1}$ also holds for our problem for non-uniform unit costs. 

\textbf{Tightness.}
For this part, the following instance shows that we cannot do better than $\frac{e}{e-1}$. 
Consider the following i.i.d. instance, for $k=1 \dots n$
\[X_k = \begin{cases}
    n & \text{w.p. } 1-\frac{1}{n} \\ 
    0 & \text{ w.p. } \frac{1}{n}.
\end{cases}\]

Since the instance is i.i.d., it does not matter if we can select the order or not, and we know that the optimal online algorithm is Weitzman~\cite{Weit1979} (also described in Algorithm~\ref{alg:weitzman_main}). We calculate the reservation value for all variables\footnote{The reservation value is the same for all variables since they are i.i.d. with the same cost.} using Equation~\ref{eq:reservation_value} to be $\s=n$. Weitzman's algorithm is therefore indifferent between stopping at the first variable and trying to find a $0$, and assume without loss of generality that it chooses to play. The expected cost of the gambler then is $ALG = \E{}{\weitz} = n$. 


        
  We calculate the prophet's expected cost, using $q=\lp(1-\frac{1}{n}\rp)$ for simplicity. Let $k \geq 1$, and call $T \in \left\{1\dots,n\right\}$ the random time for which the first 0 appears; if there is no such a 0, set $T=n+1$. 
We have for $1\leq i \leq n$ 
\begin{equation}\label{eq:worst_case_distr}
          \Pr{}{T=i}= \frac{1}{n} q^{i-1}.
     \end{equation}
%
We now calculate the cost of the prophet
     \begin{align*}
         \text{Ski-prophet} &= \sum_{i=1}^{n+1} \Pr{}{T=i} \min(i, n+1) & \\
         & = \sum_{i=1}^{n} i\cdot \Pr{}{T=i} + (n+1)\cdot \sum_{i=n+1}^{n+1} \Pr{}{T=i} & \\
         & = \sum_{i=1}^n \frac{i}{n}q^{i-1} +\frac{n+1}{n} \sum_{i=n+1}^{+\infty} q^{i-1}& \text{Using Eq~\eqref{eq:worst_case_distr}}\\ 
         & =  n(1-2q^n) +  (n+1) q^{n}\\
         & = (n-1) \lp( 1- q^n\rp)+1 
     \end{align*}
\noindent
where in the fourth equality we used again equations~\eqref{eq:sum_i_qi}  and \eqref{eq:sum_qi} below for $k=n$.
     where in the third inequality we used the following equations: 
    \begin{align}
         \sum_{i=1}^{k} i q^{i-1} &= \frac{1+q^{k}(kq-k-1)}{(1-q)^2} \label{eq:sum_i_qi}\\
         \sum_{i=k+1}^n q^{i-1} & = \frac{q^{k}}{1-q}\label{eq:sum_qi}
        \end{align}

Putting it all together we have that
\begin{eqnarray*}
\frac{ALG}{Prophet} \geq \frac{n}{(n-1) \lp( 1- q^n\rp)+1}
\end{eqnarray*}
The right-hand side term goes to $(1-e^{-1})^{-1}$ when $n$ tends to infinity, which proves the result. 

\textbf{A weaker prophet.} Even if we compete against a  weaker prophet who knows the values of the variables but not the order of arrival, we still cannot do better than $\frac{e}{e-1}$. The weak prophet stops at the first variable if there is no $0$, which happens with probability $\lp(1-\frac{1}{n}\rp)^n \approx \frac{1}{e}$, and tries to find a $0$ otherwise. In the case there is no $0$ Weitzman overpays by $n$. The cost of the prophet then is 
\[\text{Weak-Prophet} = \E{}{\weitz} - \frac{n}{e} = \E{}{\weitz}\lp( 1- \frac{1}{e}\rp).\]
Therefore $\E{}{\weitz} \approx 1.58\cdot \text{Weak-Prophet}$
\end{proof}

\section{Conclusion/Final Remarks}
We make some final observations on the results presented.
\begin{itemize}
    \item \textbf{Remark 1}: In the minimization version, being able to turn back and take an earlier option makes the approximation ratio from arbitrarily bad (upper half of Table~\ref{table:minimization_summary}) to constant (bottom half of Table~\ref{table:minimization_summary}).
    \item \textbf{Remark 2}: Example~\ref{ex:ski_rental} show that the gap between the online optimal (i.e. Weitzman) and a prophet who knows both the order and the values or a prophet who only knows the values is $1.58$. Generally the gaps between the online optimal and the prophet are not easy to characterize (e.g. see \cite{NiazSabeSham2018,DuttGergRezvTengTsig2023}) since the online optimal usually has a complex structure. 
    
\end{itemize}

\bibliographystyle{alpha}
\bibliography{bibliography}

\end{document}